\documentclass[amsmath,amsfonts,prl,aps,onecolumn,twoside,superscriptaddress]{revtex4}

\usepackage{amsfonts}
\usepackage{amsmath}
\usepackage{graphics}
\usepackage{graphicx}
\usepackage{enumerate}

\usepackage{amssymb}
\usepackage{amsthm}

\newcommand{\appref}[1]{\hyperref[#1]{{Appendix~\ref*{#1}}}}
\newcommand{\be}{\begin{eqnarray} \begin{aligned}}
\newcommand{\ee}{\end{aligned} \end{eqnarray} }
\newcommand{\benn}{\begin{eqnarray*} \begin{aligned}}
\newcommand{\eenn}{\end{aligned} \end{eqnarray*}}
\newcommand*{\cA}{\mathcal{A}}

\newcommand*{\cE}{\mathcal{E}}

\newcommand{\bc}{\begin{center}}
\newcommand{\ec}{\end{center}}

\newtheorem{theorem}{Theorem}
\newtheorem{lemma}{Lemma}

\usepackage{amsfonts}

\def\01{\{0,1\}}

\begin{document}

\title{The classical capacity of quantum thermal noise channels to within 1.45 bits}

\author{Robert K\"onig}
\affiliation{IBM TJ Watson Research Center, 1101 Kitchawan Road, Yorktown Heights, NY 10598, USA}
\affiliation{Institute for Quantum Computing and Department of Applied Mathematics, University of Waterloo, Waterloo, ON, Canada}
\author{Graeme Smith}
\affiliation{IBM TJ Watson Research Center, 1101 Kitchawan Road, Yorktown Heights, NY 10598, USA}
\date{\today}

\begin{abstract}
We find a tight upper bound for the classical capacity of quantum thermal noise channels that is within $1/\ln 2$ bits
of Holevo's lower bound.  This lower bound is achievable using
unentangled, classical signal states, namely displaced coherent states.
Thus, we find that while quantum tricks might offer benefits, when it comes
to classical communication they can only help a bit.

\end{abstract}

\maketitle

Thermal noise affects almost all communication systems.  Even optical systems, where thermal photons are very unlikely at
room temperature, are effectively subjected to thermal noise by the noise inherent in amplification \cite{Caves82}.  
The Additive White Gaussian Noise (AWGN) channel describes classical systems subjected to thermal noise.  The communication capacity of this
channel, found by Shannon, is a central tool in classical information theory \cite{Shannon48}.  Finding the capacity of thermal noise channels with quantum
effects taken into account has long been recognized as an important question~\cite{Gordon64}.  The central issue is whether using special signal states, e.g., entangled
or non-classical  states, can boost capacity beyond strategies involving only unentangled classical states. Holevo has computed an achievable rate for thermal 
channels using displaced coherent states \cite{Holevo97}.  The purpose of this note is to show that the ultimate capacity 
of thermal channels differs from Holevo's rate by no more than $1/\ln 2\approx 1.45$ bits. 

The input to a thermal noise channel is a bosonic mode with field quadratures~$(Q,P)$.  This input interacts with
a thermal state, resulting in the channel's output.  We can model this interaction as a beam splitter with transmissivity~$\lambda$, so that letting~$(q,p)$ be the thermal environment's quadratures, the output's quadratures are~$(\sqrt{\lambda}Q + \sqrt{1-\lambda}q,\sqrt{\lambda}P + \sqrt{1-\lambda}p)$.  When the thermal state has mean photon number~$N_E$, we denote the channel by~$\cE_{\lambda,N_E}$ (see~\cite{EW05} for more on gaussian channels).

The classical capacity of a channel is the maximum rate at which information can be transmitted from sender to receiver
with errors vanishing in the limit of many uses.  It is measured in bits per channel use.  Typically, there is a mean photon number 
(or power)  constraint on the signal states used.  Holevo has shown, by a random coding argument, that the capacity of a thermal
noise channel with transmissivity $\lambda$, environment photon number $N_E$ with signal photon-number constraint $N$ satisfies
\begin{align}\label{Eq:Holevo}
C_N(\cE_{\lambda,N_E}) \geq \left(g(\lambda N+ (1-\lambda)N_E) - g((1-\lambda)N_E)\right) \frac{1}{\ln 2},
\end{align}
where $g(x) = (x+1)\ln(x+1) - x\ln x$.  This means that there exist good communication schemes with rates approaching the
right hand side.  Indeed, Holevo's coding scheme is remarkably simple and requires only displaced coherent 
states as signals \cite{HW01}.

In contrast to classical information theory \cite{Shannon48,CoverThomas}, for most quantum channels we do not know a simple expression for 
classical capacity.  This is because of the superadditivity of Holevo information \cite{Hastings09} and intimately 
related to the potential of using entangled signal states to boost capacity.  There are, however, a few channels for which 
classical capacity can be evaluated \cite{BDS97,Shor02,King03}.  
The pure loss channel, $\cE_{\lambda,0}$, the thermal noise channel with zero environment photon-number, is one such example, 
with capacity given by $C_N(\cE_{\lambda,0}) = g(\lambda N)/\ln 2$ \cite{GGLMSY04}.

We do not know the classical capacity of the general thermal noise channel, $\cE_{\lambda, N_E}$.  There are, however, 
some upper and lower bounds known.  Two upper bounds are the entanglement assisted capacity \cite{BSST99} and maximum output
entropy, both computed in \cite{HW01}.  The gap between both of these bounds and the lower bound, Eq.~\eqref{Eq:Holevo}, can
be arbitrarily large since it grows with $N$ \cite{HW01}.   Very recently, we have found stronger bounds based on the quantum
entropy power inequality \cite{KS12a,KS12b}.  For $\lambda = \frac{1}{2}$ these bounds differ from Eq.~\eqref{Eq:Holevo} by at most
$0.06$ bits, but for $\lambda \neq \frac{1}{2}$ they are looser.  There has been some hope that the capacity is simply Eq.~\eqref{Eq:Holevo},
a possibility explored in \cite{Giovannettietal04}, but no proof has been found.  Because that work focused primarily on the single-letter
minimum output entropy \cite{KR01}, unfortunately it does not lead to bounds on the capacity.

The capacity satisfies a pipelining property, $C_N(\cE_1 \circ \cE_2) \leq C_N( \cE_2)$, where $\cE_1 \circ \cE_2$ is the 
concatenation of $\cE_1$ and $\cE_2$.  This can be seen operationally by noting that one potential strategy for
communicating via $\cE_2$ is for the receiver to immediately apply $\cE_1$, so that any rate achievable over the concatenated channel
is also achievable over $\cE_2$ alone using the same signal states.  This property leads to the method of additive extensions~\cite{AE08},
where one finds upper bounds for the capacity of a channel $\cE$ by decomposing it as $\cE = \cE_1 \circ \cE_2$, where $\cE_2$ has an easily 
computed capacity.  We apply this to the thermal noise channel, which can be decomposed as
\begin{align}
\cE_{\lambda,N_E} &= \cA_{G} \circ \cE_{\tilde{\lambda},0}\ ,\label{eq:compositiongaintrans}
\end{align} where $\cA_G$ is an amplification channel with gain~$G=(1-\lambda)N_E+1$, and~$\cE_{\tilde{\lambda},0}$ is the pure loss channel with transmissivity~$\tilde{\lambda}  = \lambda/G$~\footnote{Since the channels~$\cA_G$ and $\cE_{\lambda,N_E}$ are gaussian, Eq.~\eqref{eq:compositiongaintrans} can be verified by considering their action~$\Gamma\mapsto G\Gamma+(G-1)I$ and $\Gamma\mapsto \lambda\Gamma+(1-\lambda)(2N_E+1)I$ on covariance matrices~$\Gamma$.}. Because~$C_N(\cE_{\tilde{\lambda},0})=g(\tilde{\lambda}N)/\ln 2$ is known~\cite{GGLMSY04}, this leads to the bound
\begin{align}\label{Eq:AE}
C_N(\cE_{\lambda,N_E}) \leq g\left(\frac{\lambda N}{(1-\lambda)N_E +1}\right)\cdot\frac{1}{\ln 2}\ .
\end{align}
The upper bound in Eq.~\eqref{Eq:AE} is remarkably tight.  Indeed, 
we find the following theorem.

\begin{theorem}\label{Eq:Main}
Let $\cE_{\lambda,N_E}$ be the thermal noise channel with transmissivity~$\lambda$ and environment photon number~$N_E$. Consider its classical capacity $C_N(\cE_{\lambda,N_E})$ subject to the signal photon-number constraint~$N$. Let 
\begin{align*}
\gamma(\lambda,N_E,N):=\left(g(\lambda N+ (1-\lambda)N_E) - g((1-\lambda)N_E)\right) \frac{1}{\ln 2}\ 
\end{align*}
denote the rate achievable by coding with displaced coherent states. Then
\begin{align*}
\gamma(\lambda,N_E,N)\leq C_N(\cE_{\lambda,N_E})\leq \gamma(\lambda,N_E,N)+1/\ln 2\ .
\end{align*}
\end{theorem}
To prove this theorem, we show that the difference between the
 upper bound Eq.~\eqref{Eq:AE} and Holevo's lower bound~\eqref{Eq:Holevo} does not exceed~$1/\ln 2$. This 
is an immediate consequence of Eq.~\eqref{it:upperboundproperty} of the following lemma, applied with $X=\lambda N$ and $Y=(1-\lambda)N_E$.
In fact, get the slightly stronger result that the gap is no more 
than $(1-\lambda)N_E\ln\left(1 + ((1-\lambda)N_E)^{-1}\right)/\ln 2$.

\begin{lemma}
Let $g(x)=(x+1)\ln (x+1)-x\ln x$. For $Y>0$, define the function
\begin{align*}
\Delta_Y(X)=g(X(Y+1)^{-1})-g(X+Y)+g(Y)\ .
\end{align*}
Then
\begin{enumerate}[(i)] 
\item  $\lim_{X\rightarrow\infty} \Delta_Y(X) = Y \ln (1 + Y^{-1})<1$. \label{it:limit}
\item $\Delta'_Y(X)>0$ for all $X> 0$.\label{it:derivativepositive}
\end{enumerate}
In particular, 
\begin{align}
\Delta_Y(X)< Y \ln (1 + Y^{-1})<1\qquad\textrm{ for all }X,Y> 0\ .\label{it:upperboundproperty}
\end{align}
\end{lemma}
\begin{proof}
Using~$g(x)=\ln(x+1)+1+O(1/x)$, one immediately gets
\begin{align*}
\Delta_Y(X)&=\ln\left(X(Y+1)^{-1}+1\right)-\ln(X+Y+1)+g(Y)+O(1/X)\\
&\rightarrow -\ln (Y+1)+g(Y)=Y\ln(1+Y^{-1})\qquad\textrm{ for }X\rightarrow\infty\ .
\end{align*}
Statement~\eqref{it:limit} then follows from the fact that $\ln(1+\epsilon) \leq |\epsilon|$.
Similarly, with $g'(x)=\ln(1+x^{-1})$, we obtain
\begin{align}
\Delta_Y'(X)&=(Y+1)^{-1}\ln(1+(Y+1)X^{-1})-\ln(1+(X+Y)^{-1})\rightarrow 0\qquad\textrm{ for }X\rightarrow\infty\ .\label{eq:limitderivativezero}
\end{align}
Finally, we compute the second derivative of $\Delta_Y(\cdot)$ using $g''(x)=-(x(x+1))^{-1}$. Simple algebra gives
\begin{align*}
\Delta_Y''(X)&=\frac{1}{X+Y+1}\left(\frac{1}{X+Y}-\frac{1}{X}\right)< 0\qquad\textrm{ for all }X>0\ ,
\end{align*}
which shows that~$\Delta'_Y(\cdot)$ is decreasing. With Eq.~\eqref{eq:limitderivativezero}, this implies~\eqref{it:derivativepositive}.
\end{proof}

{\it Acknowledgments---} We are grateful to Mark Wilde for comments on the manuscript.  We were both 
supported by DARPA QUEST program under contract no.HR0011-09-C-0047.

\bibliographystyle{apsrev}

\end{document}